\documentclass[11pt]{article}

\usepackage{amsmath, amsfonts, amsthm}
\usepackage{graphicx}

\newcommand{\R}{{\mathbb R}}
\newcommand{\HD}{{\mathrm{HD}}}
\newcommand{\eps}{\varepsilon}
\newcommand{\conv}{{\mathrm{conv}}}

\newtheorem{theorem}{Theorem}
\newtheorem{lemma}[theorem]{Lemma}
\newtheorem{problem}{Problem}
\newtheorem{observation}[theorem]{Observation}

\title{A variant of the Hadwiger--Debrunner $(p,q)$-problem\\ in the plane}
\date{}
\author{Sathish Govindarajan\thanks{\texttt{gsat@csa.iisc.ernet.in}. Indian Institute of Science, Bangalore, India.} \and Gabriel Nivasch\thanks{\texttt{gabrieln@ariel.ac.il}. Ariel University, Ariel, Israel.}}

\begin{document}
\maketitle

\begin{abstract}
Let $X$ be a convex curve in the plane (say, the unit circle), and let $\mathcal S$ be a family of planar convex bodies, such that every two of them meet at a point of $X$. Then $\mathcal S$ has a transversal $N\subset\R^2$ of size at most $1.75\cdot 10^9$.

Suppose instead that $\mathcal S$ only satisfies the following ``$(p,2)$-condition": Among every $p$ elements of $\mathcal S$ there are two that meet at a common point of $X$. Then $\mathcal S$ has a transversal of size $O(p^8)$. For comparison, the best known bound for the Hadwiger--Debrunner $(p, q)$-problem in the plane, with $q=3$, is $O(p^6)$.

Our result generalizes appropriately for $\R^d$ if $X\subset \R^d$ is, for example, the moment curve.
\end{abstract}

\section{Introduction}

Let $\mathcal S$ be a family\footnote{Throughout this paper we allow $\mathcal S$ to be a multi-set; meaning, the elements of $\mathcal S$ need not be pairwise distinct.} of convex bodies in $\R^d$. We say that $\mathcal S$ satisfies the \emph{$(p,q)$-condition}, for positive integers $p\ge q$, if among every $p$ elements of $\mathcal S$ there are $q$ that meet at a common point. Hadwiger and Debrunner~\cite{HD}, in their celebrated problem, asked whether a family $\mathcal S$ that satisfies the $(p,q)$-condition, for $p\ge q\ge d+1$, has a transversal of size bounded by a constant $\HD_d(p,q)$ that depends only on $d$, $p$, and $q$. (A \emph{transversal} for $\mathcal S$ is a set $N\subset\R^d$ that intersects every element of $\mathcal S$.)

This problem is a generalization of Helly's theorem~\cite{Helly}: Helly's theorem states that, if every $d+1$ elements of $\mathcal S$ intersect, then they all intersect; or, in other words, $\HD_d(d+1,d+1)=1$.

It is clear that $q$ cannot be smaller than $d+1$, since a family of $n$ hyperplanes in general position provides a counterexample: Every $d$ hyperplanes intersect, and yet a transversal must contain at least $n/d$ points.

Hadwiger and Debrunner~\cite{HD} showed that, for $q > 1+\allowbreak(d-1) p / d$, one has $\HD_d(p,q)=p-q+1$.

Alon and Kleitman~\cite{AK} settled the general question in the affirmative, by tackling the hardest case $q=d+1$. Their proof uses an impressive array of tools from discrete geometry, including the fractional Helly theorem, linear-programming duality, and weak epsilon-nets. (Alon and Kleitman later published a more elementary proof in \cite{AK_easy}.)

\paragraph{Fractional Helly.} The fractional Helly theorem~\cite{Kalai_frac,KL} (see also~\cite[pp. 195]{mat_DG}) states that, if $\mathcal S$ is a family of $n$ convex bodies in $\R^d$ such that at least an $\alpha$-fraction of the $\binom{n}{d+1}$ $(d+1)$-tuples intersect, then there exists a point $z\in\R^d$ contained in at least $\beta n$ bodies, for some $\beta>0$ that depends only on $d$ and $\alpha$. The bound $\beta \ge \alpha/(d+1)$ is asymptotically optimal for small $\alpha$.

\paragraph{Weak epsilon nets.} Given a finite point set $P\subset\R^d$ and a parameter $0<\eps<1$, a \emph{weak $\eps$-net} for $P$ (with respect to convex sets) is a set $N\subset\R^d$ that intersects every convex set that contains at least an $\eps$-fraction of the points of $P$. Alon et al.~\cite{ABFK} showed that $P$ always has a weak $\eps$-net of size bounded only by $d$ and $\eps$. The best known bounds for the size of weak $\eps$-nets are $f_2(\eps) = O(\eps^{-2})$ in the plane~\cite{ABFK, CEGGSW}, and $f_d(\eps) = O(\eps^{-d} \mathrm{polylog}(1/\eps))$ for dimension $d\ge 3$~\cite{CEGGSW, MW}.

For point sets $P$ that satisfy additional constraints, better bounds are known. For example, if $P\subset X$ for some convex curve $X\subset \R^2$, then $P$ has a weak $\eps$-net of size $O((1/\eps)\alpha(1/\eps))$, where $\alpha(n)$ denotes the very slow-growing inverse-Ackermann function (Alon et al.~\cite{AKNSS}).

Regarding lower bounds, Bukh et al.~\cite{BMN_weak} constructed, for every $d$ and $\eps$, a point set $P\subset\R^d$ for which every weak $\eps$-net has size $\Omega((1/\eps) \log^{d-1}(1/\eps))$.

\paragraph{Back to the Hadwiger--Debrunner problem.} The argument of Alon and Kleitman~\cite{AK} yields $\HD_d(p, d+1) \le f_d(c_d p^{-(d+1)})$, where $f_d$ is the upper bound for weak epsilon-nets, and $c_d>0$ is some constant. Thus, for the planar case we obtain $\HD_2(p, 3)  = O(p^6)$. 

The lower bound $\HD_d(p, d+1) = \Omega(p \log^{d-1} p)$ follows from the lower bound for weak epsilon-nets: Let $P\subset\R^d$ be a point set realizing the lower bound for weak $\eps$-nets. Let $\mathcal S$ be the set of all convex hulls of at least an $\eps$-fraction of the points of $P$. Then $\mathcal S$ satisfies the $(p,d+1)$-condition for $p = 1+d/\eps$; and every transversal for $\mathcal S$ is a weak $\eps$-net for $P$.

\paragraph{Related work.}
Many variants of the $(p,q)$-problem have been studied; see for example the survey~\cite{Eckhoff}.

Regarding the case $q=2$, Danzer~\cite{danzer} (answering a question of Gallai) showed that any family of pairwise intersecting disks in the plane (i.e., satisfying the $(2,2)$-condition) has a transversal of size $4$, and that this bound is optimal.
More generally, Gr\"unbaum~\cite{grunbaum} showed that any family of pairwise intersecting homothets of a fixed convex body in $\R^d$ has a transversal bounded in terms only of $d$.

Kim et al.~\cite{kim}, together with Dumitrescu and Jiang~\cite{DJ}, showed that for homothets of a convex body in $\R^d$
having the $(p,2)$-property, the transversal number is at most $c_d p$ for some constants $c_d$.

\subsection{Our variant of the problem}

We asked ourselves the following question: Can we obtain smaller transversals for $\mathcal S$ if we impose an additional constraint in $\mathcal S$, analogous to the convex-curve constraint for weak epsilon-nets?

In this spirit, we raised the following problem: Let $X$ be a convex curve in the plane (say, $X$ could be the unit circle $\{(x, y) \mid x^2+y^2=1\}$). Let $\mathcal S$ be a family of planar convex bodies as before. We now strengthen the $(p,q)$-condition by requiring that, among every $p$ elements of $\mathcal S$, at least $q$ meet at a point of $X$. What can we say then about the minimum size of a transversal for $\mathcal S$?

\begin{problem}
Let $X$ be a convex curve in the plane, and let $\mathcal S$ be a family of planar convex bodies, such that among every $p$ elements of $\mathcal S$, three of them meet at a point of $X$. Then we know that $\mathcal S$ has a transversal of size $\HD_2(p,3) = O(p^6)$. Does $\mathcal S$ have a smaller transversal?
\end{problem}

Since the conterexample that required $q\ge 3$ does not hold in this new setting, we can ask what happens when $q=2$.

\begin{problem}
Now suppose that among every $p$ elements of $\mathcal S$, two of them meet at a point of $X$. Does $\mathcal S$ then have a transversal of size depending only on $p$?
\end{problem}

We do not know the answer to the first question, but we answer the second question in the affirmative:

\begin{theorem}\label{thm_hd_conv}
Let $X$ be a convex curve in the plane, and let $\mathcal S$ be a family of planar convex bodies. Then:
\begin{enumerate}
\item[(a)] If every pair of elements of $\mathcal S$ meet at a point of $X$, then there exists a point $z\in\R^2$ that intersects at least a $1/15800$-fraction of the elements of $\mathcal S$, and $\mathcal S$ has a transversal of size at most $1.75\cdot 10^9$.

\item[(b)] If among every $p$ elements of $\mathcal S$, two of them meet at a point of $X$, then there exists a point $z\in\R^2$ that intersects a $\Omega(p^{-4})$-fraction of the elements of $\mathcal S$, and $\mathcal S$ has a transveral of size $O(p^8)$.
\end{enumerate}
\end{theorem}

A generalization of Theorem~\ref{thm_hd_conv} for $\R^d$ is discussed in Section~\ref{sec_Rd}.

\section{The proof}

The first step (for case (\emph{b}) only) is to apply Tur\'an's theorem~\cite{turan} (see also~\cite{thebook}):

\begin{lemma}\label{lemma_many_pairs}
Let $X$ be a convex curve in the plane, and let $\mathcal S$ be a family of $n$ planar convex bodies, such that among every $p$ elements of $\mathcal S$, two of them meet at a point of $X$. Then, the number of pairs of elements of $\mathcal S$ that meet at a point of $X$ is at least $n^2/(2p)$.
\end{lemma}

\begin{proof}
Let $G$ be a graph containing a vertex for every element of $\mathcal S$, and containing an edge for every pair of elements that \emph{do not} meet at any point of $X$. Then our assumption on $\mathcal S$ is equivalent to saying that $G$ contains no clique of size $p$. Therefore, by Tur\'an's theorem, $G$ contains at most $\left(1 - \frac{1}{p-1}\right)\frac{n^2}{2}$ edges, so it is missing more than $n^2/(2p)$ edges.
\end{proof}

The second, and main, step is to prove a fractional-Helly-type lemma:

\begin{lemma}\label{lemma_frac_H_type}
Let $X$ be a convex curve in the plane, and let $\mathcal S$ be a family of $n$ planar convex bodies. Then:
\begin{itemize}
\item[(a)] If every pair of elements in $\mathcal S$ meet at a point of $X$, then there exists a point $z\in\R^2$ that is contained in at least $n/15800$ elements of $\mathcal S$.
\item[(b)] If a $\gamma$-fraction of the $\binom{n}{2}$ pairs of elements of $\mathcal S$ meet at a point of $X$, for some $0<\gamma<1$, then there exists a point $z\in\R^2$ that is contained in at least $\Omega(\gamma^4 n)$ elements of $\mathcal S$.
\end{itemize}
\end{lemma}

\begin{proof}
Let $S_1, S_2, \ldots, S_n$ be the objects in $\mathcal S$. We think of each set $S_i$ as ``colored" with color $i$. For each pair $S_i$, $S_j$ that meet at a point of $X$, select a point $p_{ij} \in X \cap S_i \cap S_j$. In case (\emph{a}) we have $N=\binom{n}{2}$ points $p_{ij}$, while in case (\emph{b}) we have $N=\gamma \binom{n}{2}$ points. Note that these points are not necessarily pairwise distinct (in fact they could all be the same point); however, that would only make our problem easier.

Sort the points $p_{ij}$ in weakly circular order around $X$, and rename the sorted points $Q = (q_0, q_1, \ldots, q_{N-1})$. We treat $Q$ as a circular list, so after $q_{N-1}$ comes $q_0$. Each $q_a$ is colored with two distinct colors among $1, \ldots, n$ (corresponding to the two objects in $\mathcal S$ that defined $q_a$), and each pair of colors occurs at most once (or exactly once in case (\emph{a})).

Let $Y = (y_0, \ldots, y_{N-1}) \subset X$ be a circular list of ``separator" points, such that $y_i$ lies (weakly) between $q_{i-1}$ and $q_i$ for every $i$.

Note that each quadruple of separator points $y=(y_a, y_b, y_c, y_d)$ (listed in circular order) defines a partition of $Q$ into four intervals: $[q_a, q_{b-1}]$, $[q_b, q_{c-1}]$, $[q_c, q_{d-1}]$, $[q_d, q_{a-1}]$. The quadruple $y$ is said to ``pierce" color $i$ if each of these four intervals contains a point colored with color $i$.

We make use of the following observation, which was previously used in \cite{AKNSS} and \cite{CEGGSW}.

\begin{observation}\label{obs_pierce}
Let $y=(y_a, y_b, y_d, y_d)$ be a quadruple of separator points, and let $z\in\R^2$ be the point of intersection of segments $y_ay_c$ and $y_by_d$. Then, if $y$ pierces color $i$, then $z\in S_i$ (see Figure~\ref{figs} (left)).
\end{observation}

Our strategy is to show that a randomly-chosen quadruple of separators pierces, in expectation, a constant fraction of the colors.

Define the \emph{distance} between two points $q_a, q_b \in Q$ as $\min{\{(b-a)\bmod N, (a-b)\bmod N\}}$.

We now choose a parameter $\alpha<1$ independent of $n$: For case (\emph{a}) we set $\alpha = 0.027$, while for case (\emph{b}) we set $\alpha = \gamma/300$. We call a color $i$ \emph{spread out} if there exist four points $q_a, q_b, q_c, q_d\in Q$, colored with color $i$, such that all the pairwise distances between these four points are at least $\alpha N$.

\begin{observation}\label{obs_random}
A randomly-chosen quadruple $y = (y_a, y_b, y_c, y_d)$ has probability at least $24\alpha^3 (1-3\alpha)$ of piercing a given spread-out color.
\end{observation}

\begin{proof}
Suppose color $i$ is spread out. Consider four points $q_a, q_b, q_c, q_d\in Q$ in cyclic order, that prove that $i$ is spread out. Let the distances between them in cyclic order be $\beta_1 N$, $\beta_2 N$, $\beta_3 N$, $\beta_4 N$; so $\beta_1 + \cdots + \beta_4=1$. Then $y$ pierces color $i$ with probability at least $24\beta_1 \beta_2\beta_3 \beta_4$. Subject to the constraints $\beta_i \ge \alpha$ for $1\le i \le 4$, this quantity is minimized when $\beta_1 = \beta_2 = \beta_3 = \alpha$ and $\beta_4 = 1-3\alpha$.
\end{proof}

We now proceed to derive a lower bound on the number of spread-out colors.

First, we characterize when a color is spread out:

\begin{observation}\label{obs_3intervals}
For each color $i$, exactly one of the following two options holds:
\begin{enumerate}
\item Color $i$ is spread out.
\item All the instances of color $i$ occur in at most three intervals of $Q$, each of length at most $\alpha N$.
\end{enumerate}
\end{observation}

\begin{proof}
If the second condition is true then clearly color $i$ is not spread out, because we can at most choose $q_a$, $q_b$, and $q_c$ from three different intervals, and then we have no way of choosing $q_d$.

For the other direction, suppose color $i$ is not spread out. Let $I$ be the longest interval in $Q$ that is completely free of color $i$. We must certainly have $|I| > \alpha N$, since otherwise $Q$ would have $1/(2\alpha)$ points in cyclic order, with pairwise distances at least $\alpha N$, all colored with color $i$; and $1/(2\alpha)>4$.

To the left and right of $I$ are points $q_a$ and $q_b$, respectively, colored with color $i$. Let $q_{a'}$ be the farthest point left of $q_a$, still within distance $\alpha N$ of $q_a$, that is colored with color $i$. Similarly, let $q_{b'}$ be the farthest point right of $q_b$, still within distance $\alpha N$ of $q_b$, that is colored with color $i$. Let $q_c$ be the first point left of $q_{a'}$ colored with color $i$; and let $q_d$ be the first point right of $q_{b'}$ colored with color $i$.

\begin{figure}
\centerline{\includegraphics{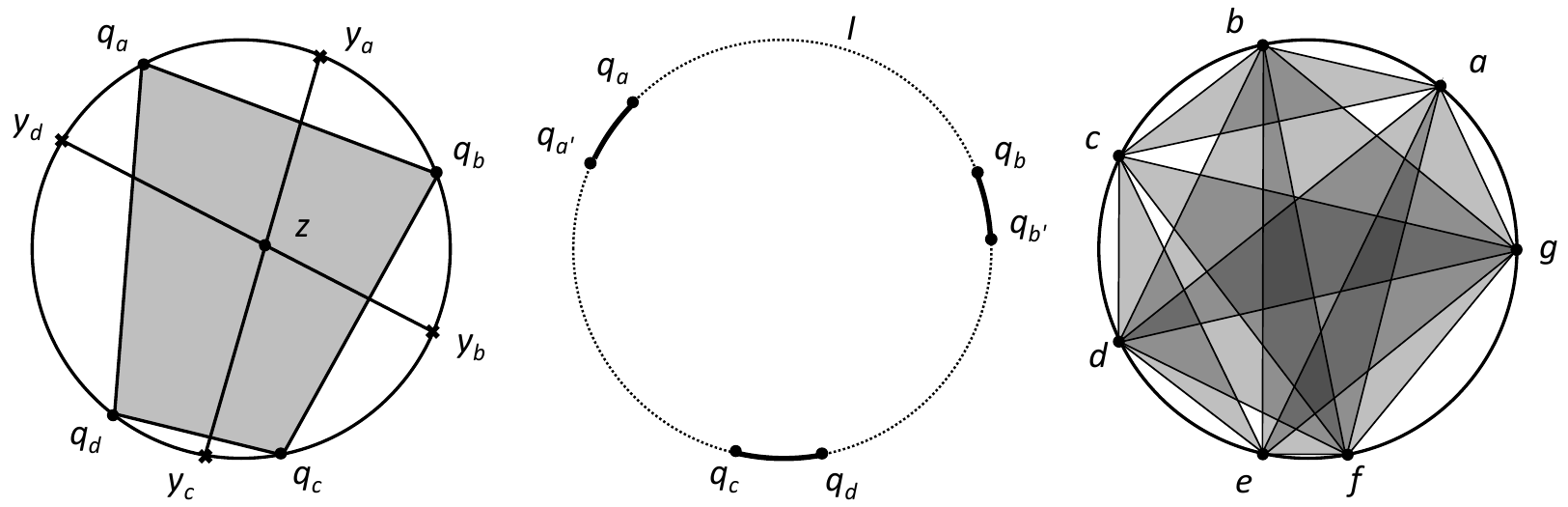}}
\caption{\label{figs}Left: If the quadruple of separators pierces color $i$, then $z\in S_i$. Center: If color $i$ is not spread out, then it is contained in three small intervals. Right: A family that requires a transversal of size $3$.}
\end{figure}

The distance between $q_c$ and $q_d$ must be less than $\alpha N$, since otherwise $q_c$, $q_a$, $q_b$, $q_d$ would prove that color $i$ is spread out. Thus, all instances of color $i$ are contained in the intervals $[q_d, q_c]$, $[q_{a'}, q_a]$, $[q_b, q_{b'}]$. See Figure~\ref{figs} (center).
\end{proof}

We now derive an upper bound on the number of colors that are \emph{not} spread out.

\begin{lemma}\label{lemma_not_spread_bd}
In case (a) the number of colors that are not spread out is at most $3\sqrt{3\alpha}n + o(n)$. In case (b) this number is at most $(1-\gamma/4)n + o(n)$.
\end{lemma}

\begin{proof}
We will use the following graph-theoretic observation:

\begin{observation}\label{obs_graph}
Let $G=(V, E)$ be a graph, and for every $v\in V$ let $g(v) = \sum_{w\in N(v)} d(w)$ denote the sum of the degrees of the neighbors of $v$.
Then, there exists a vertex $v\in V$ for which $g(v) \ge 4|E|^2/|V|^2$.
\end{observation}

\begin{proof}
We have $\sum_{v\in V} g(v) = \sum_{v\in V} d^2(v)$, since each vertex $v$ contributes exactly $d(v)$ to $d(v)$ different terms of $\sum g(v)$. Therefore, the claim follows by the Cauchy--Schwarz inequality, noting that $\sum_{v\in V} d(v) = 2|E|$.
\end{proof}

Let $m = kn$ be the number of colors that are not spread out. In case (\emph{b}) we may assume that $k>1-\gamma/2$, since otherwise we are done.

Assume for simplicity that the non-spread-out colors are $1, \ldots, m$. For each color $i$, $i\le m$, let $I_{i1}, I_{i2}, I_{i3}$ be the three intervals $Q$, of length at most $\alpha N$, on which color $i$ occurs, according to Observation~\ref{obs_3intervals}.

Let $G=(V,E)$ be a graph with $3m$ vertices, labeled $v_{ia}$ with $1\le i\le m$ and $1\le a\le 3$, and with an edge connecting vertices $v_{ia}$ and $v_{jb}$ if and only there is a point $p_{ij}$, having the pair of colors $i$ and $j$, lying on the intervals $I_{ia}$ and $I_{jb}$. In case (\emph{a}) we have $|E|=\binom{m}{2}$.

In case (\emph{b}) we have $|E| \ge N - (1-k)n^2$, since the number of points $p_{ij}$ that have a spread-out color (meaning, that $i>m$ or $j>m$) is at most $n(n-m) = (1-k)n^2$. Thus, $|E| \ge (k + \gamma/2-1)n^2$ ignoring lower-order terms. Note that this quantity is positive, by our assumption on $k$.

Denote by $g(v) = \sum_{w\in N(v)} d(w)$ the sum of the degrees of the neighbors of vertex $v\in V$. By Observation~\ref{obs_graph}, there exists a vertex $v_{ia}$ for which $g(v_{ia}) \ge 4|E|^2/|V|^2$.

Consider the interval $I_{ia}$ corresponding to this vertex $v_{ia}$. Recall that $I_{ia}$ has length at most $\alpha N$. Let $I'$ be an interval of $Q$ of length $3\alpha N$ centered around $I_{ia}$. All the intervals $I_{jb}$ that correspond to neighboring vertices $v_{jb} \in N(v_{ia})$ lie in $I'$. Each such $I_{jb}$ contains $d(v_{jb})$ points colored with color $j$. Thus, $I'$ contains at least $g(v_{ia})$ ``colorings" of points. But at most two ``colorings" happen at each point, so $|I'| \ge g(v_{ia})/2 \ge 2|E|^2/|V|^2$.

Therefore, $3\alpha N \ge 2|E|^2/|V|^2$. In case (\emph{a}) we substitute $|V| = 3m$, $|E| \approx m^2/2$, and $N \approx n^2/2$ (ignoring lower-order terms); we obtain $m \le 3\sqrt{3\alpha}n + o(n)$, as claimed.

In case (\emph{b}) we substitute $|V| = 3kn$, $|E| \ge (k+ \gamma/2 -1)n^2$, $N = \gamma n^2/2$, and $\alpha = \gamma/300$. Solving for $k$, we obtain
\begin{equation*}
k\le \frac{1-\gamma/2}{1-3\gamma/20}.
\end{equation*}
Since $0<\gamma<1$, this quantity is at most $1-\gamma/4$, completing the proof.
\end{proof}

Thus, the number of spread-out colors is at least $(1-3\sqrt{3\alpha})n - o(n)$ in case (\emph{a}), and $\Omega(\gamma n)$ in case (\emph{b}).

To conclude the proof of Lemma~\ref{lemma_frac_H_type}, we put together Observation~\ref{obs_random} and Lemma~\ref{lemma_not_spread_bd}. They give us a lower bound on the expected number of colors that are pierced by a randomly-chosen quadruple of separators. There must exist a quadruple $y=(y_a, y_b, y_c, y_d)$ that achieves this expectation.

In case (\emph{a}), the expectation is $24\alpha^3(1-3\alpha)(1-3\sqrt{3\alpha}) n - o(n)$. Since we chose $\alpha = 0.027$ (which is close to optimal), this is at least $n/15800$ for large enough $n$.
 
For case (\emph{b}) we note that the bound in Observation~\ref{obs_random} is $\Omega(\alpha^3)$, which is $\Omega(\gamma^3)$ by our choice of $\gamma$. Hence, $y$ pierces $\Omega(\gamma^4 n)$ colors.

In both cases, by Observation~\ref{obs_pierce}, the point of intersection $z = y_ay_c \cap y_by_d$ is the desired point. This completes the proof of Lemma~\ref{lemma_frac_H_type}.
\end{proof}

The final step is to apply the standard Alon--Kleitman machinery. We follow Matou\v sek's presentation in \cite{mat_DG}:

\begin{proof}[Proof of Theorem~\ref{thm_hd_conv}]
We recall some concepts. Given a finite family $\mathcal S$ of objects in $\R^d$, a \emph{fractional transversal} for $\mathcal S$ is a finite point set $N\subset\R^d$, together with a weight function $w:N\to [0,1]$, such that $\sum_{x\in N\cap S} w(x) \ge 1$ for each $S\in \mathcal S$. (A regular transversal is then a fractional transversal for which $w(x) = 1$ for all $x\in N$.) The size of the fractional transversal is defined as $\sum_{x\in N} w(x)$.

A \emph{fractional packing} for $\mathcal S$ is a weight function $\phi:\mathcal S\to [0,1]$, such that $\sum_{S\in \mathcal S:x\in S} \phi(S) \le 1$ for every point $x\in\R^d$. The size of the fractional packing is defined as $\sum_{S\in\mathcal S} \phi(S)$.

Since $\mathcal S$ has a finite number of elements, they define a partition of $\R^d$ into a finite number of regions. It does not matter which point we choose from each region, and therefore, there is only a finite number of points we have to consider.

The problems of minimizing the size of a fractional transversal of $\mathcal S$, and of maximizing the size of a fractional packing of $\mathcal S$, are both linear programs, and furthermore, they are duals of each another. Therefore, by LP duality, the size of their optimal solutions coincide (see also~\cite{LP}). We denote by $\tau^*(\mathcal S)$ the optimal size of the linear programs.

Now consider the family $\mathcal S$ given in Theorem~\ref{thm_hd_conv}. Recall that $\mathcal S$ satisfies our strengthened $(p,2)$-condition: Among every $p$ elements of $\mathcal S$, two meet at a point of $X$ (with $p=2$ in case (\emph{a})). We can assume that every element of $\mathcal S$ intersects $X$, since otherwise, the remaining elements would satisfy the $(p-1, 2)$-condition.

Let $\phi$ be a fractional packing for $\mathcal S$ achieving the optimal size $\tau^* = \tau^*(\mathcal S)$. We can assume that $\phi(S)$ is rational for every $S\in\mathcal S$. Write $\phi(S) = m(S)/D$, where $m(S)$ and $D$ are integers and $D$ is a common denominator. Then $\sum_{S\in\mathcal S} m(S) = \tau^* D$, and
\begin{equation}\label{eq_mS}
\sum_{S\in \mathcal S:x\in S} m(S) \le D \qquad \text{for every point $x\in \R^d$}.
\end{equation}

Define a family of objects $\mathcal T$ obtained by repeating each $S\in\mathcal S$ $m(S)$ times. Since $\mathcal S$ satisfies our strengthened $(p,2)$-condition, so does $\mathcal T$ (if among the $p$ elements we select two copies of the same object, then they clearly meet in $X$). Thus, by Lemmas~\ref{lemma_many_pairs} and~\ref{lemma_frac_H_type}, there exists a point $z\in \R^2$ contained in at least an $\eps$-fraction of the $\tau^* D$ objects in $\mathcal T$, where $\eps = 1/15800$ in case (\emph{a}) or $\eps = \Omega(p^{-4})$ in case (\emph{b}). On the other hand, equation (\ref{eq_mS}) implies that $z$ cannot intersect more than $D$ objects of $\mathcal T$. Hence, $\tau^* \le 1/\eps$.

By LP duality, this means that $\mathcal T$ has a fractional transversal $(N, w)$ of size at most $1/\eps$. As before, we can assume that all the weights in the fractional transversal are rational. We replace $N$ by an unweighted point set $N'$, in which each point of $x\in N$ is replaced by a tiny cloud of size proportional to $w(x)$. Then, each object in $\mathcal T$ (and thus, each object in $\mathcal S$) contains at least an $\eps$-fraction of the points of $N'$.

Finally, we take a weak $\eps$-net $M$ for $N'$. Since $M$ intersects every convex set that contains an $\eps$-fraction of the points of $N'$, $M$ is our desired transversal for $\mathcal S$. Its size is $f_2(\eps) = O(\eps^{-2})$, which in case (\emph{b}) is $O(p^8)$. For case (\emph{a}) we use the more explicit bound $f_2(\eps) \le 7\eps^{-2}$ of Alon et al.~\cite{ABFK}, and we get $|M| \le 1.75\cdot 10^9$, as claimed.\footnote{The bound of Alon et al.~can actually be improved to $f_2(\eps)\le 6.37\eps^{-2} + o(\eps^{-2})$ by simply optimizing the parameter involved in the divide-and-conquer argument. This would lead to a modest improvement in our bound for $|M|$.}
\end{proof}

\section{Generalization to $\R^d$}\label{sec_Rd}

\paragraph{Convex curves.} A \emph{convex curve} in $\R^d$ is a curve that intersects every hyperplane at most $d$ times~\cite[p.~314]{ziv}. The most well known convex curve is the \emph{moment curve}
\begin{equation*}
\bigl\{(t, t^2, \ldots, t^d) \bigm| t\in \R \bigr\}.
\end{equation*}
If $d$ is even, then a convex curve in $\R^d$ can be open (like the moment curve) or closed, like the \emph{Carath\'eodory curve}~\cite[p.~75]{ziegler}
\begin{equation*}
\bigl\{ (\sin t, \cos t, \sin 2t, \cos 2t, \ldots, \sin \tfrac{d}{2}t, \cos \tfrac{d}{2}t) \bigm| 0\le t<2\pi \bigr\}.
\end{equation*}
For $d$ even it is convenient to think of the curve as being always closed, by pretending, if necessary, that the curve's two endpoints are joined together. In other words, for $d$ odd we consider the points on the curve to be linearly ordered, while for $d$ even we consider the points to be circularly ordered.

\paragraph{Weak epsilon-nets.} The result by Alon et al.~\cite{AKNSS} on weak epsilon-nets mentioned in the introduction generalizes as follows: If $P$ is a finite point set that lies on a convex curve $X\subset \R^d$, then $P$ has a weak $\eps$-net of size at most $(1/\eps)2^{\mathrm{poly}(\alpha(1/\eps))}$.

Note that this bound is barely superlinear in $1/\eps$, and it is much stronger than the general bound for weak $\eps$-nets in $\R^d$.

\subsection{Generalization of our result} Theorem~\ref{thm_hd_conv}(\emph{b}) generalizes as follows:

\begin{theorem}\label{thm_hd_conv_Rd}
Let $X$ be a convex curve in $\R^d$, and let $\mathcal S$ be a family of convex bodies in $\R^d$, with the property that among every $p$ elements of $\mathcal S$, two meet at a point of $X$.
Then, there exists a point $z\in \R^d$ intersecting a $\Omega(p^{-j})$-fraction of the elements of $\mathcal S$, for some constant $j = d^2/2 + O(d)$.

As a result, $\mathcal S$ has a transversal of size $O(p^{j'})$ for some constant $j' = d^3/2 + O(d^2)$.
\end{theorem}

For comparison, the bound for $\HD_d(p, d+1)$ obtained by Alon and Kleitman is only $O(p^{j''})$ for $j'' = d^2 + O(d)$.

The proof of Theorem~\ref{thm_hd_conv_Rd} proceeds like the proof of Theorem~\ref{thm_hd_conv}(\emph{b}), with the following main changes:

Instead of Observation~\ref{obs_pierce} we use the following Lemma:

\begin{lemma}[Alon et al.~\cite{AKNSS}\footnote{Alon et al.~state this lemma specifically for the moment curve, but it is true for any convex curve.}]\label{lemma_jtuple}
Let $X$ be a convex curve in $\R^d$, and define
\begin{equation}\label{eq_j}
j=\begin{cases}(d^2+d+2)/2,&d \text{ even};\\ (d^2+1)/2,& d \text{ odd}.\end{cases}
\end{equation}
Let $A$ be a set of $j$ points on $X$. Note that $A$ partitions $X$ into $j+1$ intervals if $d$ is odd, or $j$ intervals if $d$ is even.

Then, there exists a point $p\in\conv(A)$ with the following property: For every set $B\subset X$ that contains a point in each of the above-mentioned intervals, we have $p\in\conv(B)$.
\end{lemma}

In our application of the Lemma, $A$ plays the role of the separator points, and $B$ plays the role of the points colored with color $i$.
Hence, instead of quadruples of separator points, we consider $j$-tuples.

A color $i$ is now \emph{spread out} if there exist $j+1$ points for $d$ odd, or $j$ points for $d$ even, colored with color $i$, such that all the pairwise between these points are at least $\alpha N$. Then, exactly one of the following is true: Either color $i$ is spread out, or all instances of color $i$ occur in at most $j$ intervals for $d$ odd, or $j-1$ intervals for $d$ even, each of length at most $\alpha N$.

The probability of a random $j$-tuple of separators piercing a spread-out color is now $\Omega(\alpha^j)$ for $d$ odd, and $\Omega(\alpha^{j-1})$ for $d$ even. Instead of setting $\alpha = \gamma/300$, we set $\alpha = c_d \gamma$ for a small enough positive constant $c_d$.

The remaining details are left to the reader.

\section{Conclusion}

Figure~\ref{figs} (right) shows a family of seven convex sets, every pair of which meet at a point of the unit circle, that requires a transversal of size $3$. The points $a, \ldots, g$ are uniformly spaced along the unit circle, except for $f$, which has been moved a bit towards $e$. The seven sets are the convex hulls of $abc$, $cde$, $efa$, $bdf$, $adg$, $beg$, $cfg$, respectively.  If $2$ points were enough to pierce all the triangles, then at least one point must intersect $4$ triangles. There are three regions which are overlaps of $4$ triangles (the darkest shades of gray in the figure). But in each case, there are three triangles left that cannot be intersected with a single point. 

We believe that the true bound for this problem is less than $10$.

\end{document}